\documentclass[authoryear,12pt,3p,letter]{jowarticle}
\usepackage{graphicx}
\usepackage{amsthm,amsmath}
\usepackage{amssymb}
\usepackage{amsfonts}
\usepackage{natbib}
\usepackage{setspace}
\usepackage[all]{xy}
\usepackage{enumitem}
\usepackage{titlesec}
\usepackage{mathrsfs}

\usepackage{epigraph}
\makeatletter
\renewcommand\section{\@startsection{section}{1}{\z@}{-3.25ex plus -1ex minus -.2ex}{1.5ex plus .2ex}{\normalsize\bf}}
\renewcommand\subsection{\@startsection{subsection}{2}{\z@}{-3.25ex plus -1ex minus -.2ex}{1.5ex plus .2ex}{\normalsize\bf}}
\renewcommand\subsubsection{\@startsection{subsubsection}{3}{\z@}{-3.25ex plus -1ex minus -.2ex}{1.5ex plus .2ex}{\normalsize\bf}}
\makeatother

\newtheorem{innercustomgeneric}{\customgenericname}
\providecommand{\customgenericname}{}
\newcommand{\newcustomtheorem}[2]{%
  \newenvironment{#1}[1]
  {%
   \renewcommand\customgenericname{#2}%
   \renewcommand\theinnercustomgeneric{##1}%
   \innercustomgeneric
  }
  {\endinnercustomgeneric}
}

\newtheorem{thm}{Theorem}

\newtheorem{prop}[thm]{Proposition}

\newtheorem{fact}{Fact}

\newcustomtheorem{prop*}{Proposition}
\newcustomtheorem{cor*}{Corollary}

\begin{document}
\begin{frontmatter}
\title{The Local Validity of Special Relativity, \\ Part 2: Matter Dynamics}
\author{Samuel C. Fletcher}\ead{scfletch@umn.edu}
\address{Department of Philosophy\\ University of Minnesota, Twin Cities}
\author{James Owen Weatherall}\ead{weatherj@uci.edu}
\address{Department of Logic and Philosophy of Science\\ University of California, Irvine}
\begin{abstract}
In this two-part essay, we distinguish several senses in which general relativity has been regarded as ``locally special relativistic''.  
In Part 1, we focused on senses in which a relativistic spacetime may be said to be ``locally (approximately) Minkowskian''.  
Here, in Part 2, we consider what it might mean to say that a matter theory is ``locally special relativistic''.  
We isolate and evaluate three criteria in the literature and show that they are incompatible: matter theories satisfying one will generally violate others.  
We then consider what would happen if any of those criteria failed for a given theory.
\end{abstract}
\end{frontmatter}

\doublespacing

\section{Introduction}\label{sec:introduction}

In Part 1 of this two-part essay, we considered interpretations of the claim that spacetime is locally flat that understand it in terms of properties of the (local) geometry of spacetime.  But often, authors invoke local flatness in the context of discussions not just of geometry, but also of the dynamics of matter.  The idea is that the local approximate flatness of spacetime is reflected in the local dynamics of matter fields within spacetime, or in other words, that not only is spacetime locally approximately flat, but matter behaves (locally, approximately) as if it ``sees'' that flat geometry.  This claim, or variations thereon, is sometimes known as the ``strong equivalence principle'' (SEP).  It suggests another possible interpretation, or class of interpretations, of the claim that spacetime is locally approximately flat, different in character from those we discussed in Part I.

\begin{description}
    \item[Matter Dynamics Interpretation:] The dynamics of matter fields in general relativity are just locally (approximately) those of matter fields in special relativity.
\end{description}

\noindent Our goal in the present Part is, first, to emphasize how this claim, as stated, is ambiguous in two ways, and second, to discuss options for resolving those ambiguities.  

The first ambiguity concerns what, precisely, it means to say that  matter dynamics are ``locally special relativistic'' (LSR).  As we discuss presently, there are several attempts to make this idea precise in the literature; these turn out to be not only different in character, but generally incompatible with one another---a fact that does not seem to have been widely appreciated.  Moreover, what is involved in checking whether a theory is, in fact, LSR on each of these different approaches is strikingly different, which means that it is not clear that a theory that appears manifestly LSR by some criteria is also LSR by other criteria, even when those criteria are, in principle, compatible.

Even if we settle on a particular understanding of when a theory is LSR, a second ambiguity persists.  
As we showed in Part I, every relativistic spacetime is locally approximately flat in a precise sense.  
But it is clear that \emph{not} every possible matter theory is LSR, by any reasonable criterion.  
This raises the question of whether the claim that matter dynamics in general relativity are LSR is a \textit{descriptive} one---``an expression of a fact about [matter] theories, that, if it holds true, is remarkable'' \citep[p.~350]{Knox2013-KNOESG}---or a \textit{constitutive} one---an interpretive supplement to matter theories that ensures that the spacetime metric represents physical quantities of length and duration \citep[p.~151]{Brown2005}.\footnote{
    See also \citet[\S4.1]{ReadBrownLehmkuhl2018}.} 
In the former case, observing that matter dynamics are LSR may be of great pragmatic value, for instance because they are simple to analyze or support certain standard inferences, but it would be of little foundational significance.  
On the other hand, if being LSR is somehow a necessary or sufficient supplement to a matter theory for it to be compatible with (or even engender) the usual interpretation of the metric in general relativity, one would like to have greater clarity on precisely what goes wrong if one introduces matter theories that \textit{fail} to be LSR.

The remainder of this paper will address each of these issues.  In section \ref{sec:preliminaries}, we will introduce some technical preliminaries that will clarify what we mean by a ``matter theory'' in the present context and which will provide the framework for the remainder of the paper.  In sections \ref{sec:syntactic}--\ref{sec:Ehlers}, we will isolate and discuss three criteria for when a matter theory should be said to be LSR.  We will argue in section \ref{sec:syntactic} that one common criterion, which we call the ``syntactic'' criterion because it is based on the ``form'' of a system of equations, is inadequate because it is hyperintensional: it renders different verdicts for systems of equations that should be seen as equivalent (because they have the same solutions).  We will then discuss two alternative ``semantic'' criteria, due, respectively, to \citet{sonego1993coupling} (section \ref{sec:Sonego}) and \citet{ehlers1973survey} (section \ref{sec:Ehlers}).   

In section \ref{sec:comparisons}, we will consider the relationship between these criteria.  We will show, in particular, that the Sonego-Faraoni and Ehlers criteria are inequivalent and, with certain natural background choices, they are also incompatible, in the sense that theories satisfying one cannot satisfy the other (and vice versa).  We will also suggest that the Ehlers criterion is more likely to give verdicts that are compatible with those expected by advocates of the syntactic criterion.  

In section \ref{sec:whoCares}, we will briefly turn to the second ambiguity.  We will argue that, once we have precisely stated what ``LSR'' might mean, it is not at all clear that any of the criteria we discuss is necessary for a reasonable or realistic matter theory in general relativity.  Indeed, one of the criteria is not satisfied by even standard realistic theories with wide applicability; we show in section \ref{sec:Ehlers} that the other criterion is satisfied by some theories, but our arguments are not readily extended to realistic non-linear theories such as Yang-Mills theory, and so it is not yet clear whether such theories should count as LSR.  Moreover, we will suggest that the somewhat non-standard theories that satisfy the Sonego-Faraoni criterion presented in section \ref{sec:Sonego} are not obviously pathological or incompatible with the standard interpretation of the spacetime metric in general relativity, even if they are not LSR.  And so it is not clear what constitutive force the requirement that a theory be LSR could have.

We will conclude in section \ref{sec:conclusion} with a brief discussion of two future directions for research.

\section{Preliminaries: What is a Matter Theory?}\label{sec:preliminaries}

To fix ideas, we suppose that matter fields in GR may be represented by (local) sections of bundles of field values over spacetime manifolds $M$, and that the dynamics of matter are given by partial differential equations acting on those sections.\footnote{For background on fiber bundles, particularly in physics, see \citet{Kobayashi+Nomizu}, \citet{Bleecker}, or \citet{Palais}. For more on this perspective regarding matter theories, see \citet{GerochPDE} and \citet{NGNF}.  In what follows, we will sometimes assume the fibers are endowed with a linear or affine structure, such that one can readily impose a norm on each fiber, or on sections, though this is not essential to matter theories.}  That is: we suppose that a ``theory of matter'' consists in (1) a systematic assignment, to each spacetime $(M,g_{ab})$, of a bundle $B\xrightarrow{\pi} M$, whose typical fiber is some fixed manifold $F$ representing possible field values at each point, with $F$ independent of the choice of spacetime $(M,g_{ab})$;\footnote{By ``systematic'' we mean, roughly, that the association between smooth manifolds and bundles of field values is functorial; see the discussion of ``natural'' bundles in \citet[\S 1.4]{NGNF}.} and (2), a system of differential equations, each of the form $P[\varphi,g_{ab}]\varphi = f$, on sections $\varphi:O\rightarrow B$, for $O\subseteq M$, of each of these bundles, where $P$ is some differential operator on sections, and both $P$ and $f$ may depend on $g_{ab}$, $\varphi$, and other fields defined on $M$.  (We explicitly note the dependence of $P$ on $\varphi$ and $g_{ab}$ for future convenience.\footnote{We make no special restrictions on the differential operator here, though all of the actual cases we consider in what follows will be at most quasilinear.})   

The first part of the definition, (1), captures the idea that the space of possible values that a given matter theory may ascribe to a point is independent of not only the point but also the ambient spacetime,\footnote{This might be interpreted as an extension of the fact discussed in Part 1 that all spacetimes with the same metric signature geometrically have the same local character: every spacetime is a universal approximating spacetime.} and also that the space of possible configurations of fields described by the theory is both fixed by the theory and may be given by a particular bundle structure.  This is a generalization for an arbitrary field of the idea that the ``electromagnetic field'' at any point of any spacetime is represented by the value, at that point, of a tensor field---namely, a two-form $F_{ab}$---on that spacetime.  The second part of the definition, (2), states that, whatever else is the case, ``matter dynamics'' are differential equations, and that any matter theory must stipulate, on any spacetime, what the equations governing this specific form of matter will be.  Presumably, any adequate treatment of matter theories will need to say more about what it means to impose the ``same'' differential equation, \emph{mutatis mutandis}, on bundles over different manifolds and with different metrics; and it will need to say more about how ``source terms'' are determined across different spacetimes.\footnote{See fn.~\ref{fn:sources} for a brief comment on this issue in the context of Maxwell theory.} But for present purposes we do not need to settle those issues.  What we presuppose here is only that once a matter theory is fixed, there is no ambiguity regarding what the field bundle over a given spacetime is or what the equations governing those fields are.
This is just an expression of the fact that a matter theory describes a class of mathematical models that, the theory asserts, represent the physically possible configurations of the theory's subject, a particular material field \citep{FletcherModality}.

To take an example, one possible matter theory, which we will call (real, mass $m$) Klein-Gordon theory, associates, with each spacetime $(M,g_{ab})$, the (unique) trivial bundle $B\xrightarrow{\pi} M$ over $M$ with typical fiber $\mathbb{R}$; and imposes on sections $\varphi:O\rightarrow B$ of those bundles the equation $P[g_{ab}]\varphi = 0$, where $P[g_{ab}] = g^{ab}\nabla_a\nabla_b + m^2$.  Here $\nabla$ is the Levi-Civita derivative operator associated with $g_{ab}$ and $m$ is a (fixed) number. Note that there are other matter theories, on our definition, whose bundle and whose associated equations may coincide with Klein-Gordon theory on some spacetimes, such as Minkowski spacetime.  But only Klein-Gordon theory makes this assignment of equations to every spacetime.

Similarly, (vacuum) Maxwell theory associates with each spacetime $(M,g_{ab})$ a (not necessarily trivial) bundle $B\xrightarrow{\pi} M$, with typical fiber a six dimensional vector space, constructed in a canonical way from the tangent bundle of $M$ so that each element of the fiber over a point $p$ is associated with an antisymmetric tensor $F_{ab}$ at $p$; and it imposes, on sections, two equations, $d_aF_{bc} = \mathbf{0}$ and $\nabla_aF^{ab} = \mathbf{0}$.\footnote{\label{fn:sources} There are a few strategies available within this framework for dealing with sources in Maxwell theory.  One is to insist that sources must be associated with their own matter theory, such as charged Klein-Gordon theory, and consider the expanded system including Maxwell fields and the other matter fields as equations on a single bundle of field values with a direct sum structure on each fiber.  Here the inhomogeneous Maxwell equation would include some term that was an algebraic function of the other field's values.  Another option is to expand the Maxwell theory bundle so that field configurations consist of pairs $F_{ab},J_a$, where $J_a$ is a covector field appearing in the inhomogeneous Maxwell equation and subject to additional equations, such as $g^{ab}\nabla_aJ_b=\mathbf{0}$.  What does not work is to consider something like ``Maxwell theory for a particular configuration of charges'', since it is not clear how to identify specific source configurations across different spacetimes.}



\section{The Syntactic Criterion}\label{sec:syntactic}

Now suppose we have some matter theory in the sense described in section \ref{sec:preliminaries}.  What does it mean to say that this theory is LSR?
One proposal often cited in the literature concerns the form of a system of equations.  For example, \citet[p.~117]{Brown2005} characterizes the ``local validity of special relativity'' as the condition that ``the laws of physics of the non-gravitational interactions \ldots take their standard special relativistic form at [spacetime point] $p$'' in coordinates in which the spacetime metric ``takes the form diag$(1,-1,-1,-1)$ and the first derivatives of its components vanish'', i.e., Lorentz normal coordinates.\footnote{See also similar remarks by \citet{brown2001origin}.  Lorentz normal coordinates are discussed in detail in Part 1.}      \citet[pp.~169--70]{Brown2005} clarifies that ``special relativistic form'' means ``that the covariance group of the equations is the inhomogeneous Lorentz group'';  \citet{ReadBrownLehmkuhl2018}, meanwhile, invoke the ``Poincar\'e invariance'' of such equations, which means that the ``form'' of the equations at a point $p$ in Lorentz normal coordinates near $p$ is invariant under the class of coordinate transformations that preserve the property of being Lorentz-normal-at-$p$.  

Similar ideas arise in other discussions of the SEP.  For instance, \citet[p.~136]{torretti1996relativity} writes that the equivalence principle demands that ``the laws of nature take the same form in [the domain of a Lorentz chart] $U_Q$ \ldots as they would when referred to an ordinary Lorentz chart in a spacetime region where gravity is absent.''  Similarly,  \citet[p.~352]{Knox2013-KNOESG} writes, ``To any required degree of approximation, given a sufficiently small region of spacetime, it is possible to find a reference frame with respect to whose associated coordinates the metric field takes Minkowskian form, and the connection and its derivatives do not appear in any of the fundamental field equations of matter.'' 

Despite being widespread in the literature, there are several problems with such explications of what it means for a system of equations to be LSR.  These stem from the fact that this way of thinking about matter dynamics is essentially ``syntactic''.  That is, it concerns how the dynamical equations \textit{appear} in a semi-formal language of coordinates, their partial derivatives, and arithmetic functions of these. Even formulations along these lines that appear to invoke symmetry properties of equations---such as Poincar\'e invariance or covariance---are syntactic, because they concern how the \textit{form} of an equation changes, as one changes coordinate systems, relative to certain conventions as to what counts as sameness of form.
But one might expect that a satisfactory characterization of what it means for a system of equations to be LSR would reflect something about not the expression of the equations, but what the equations represent---some suitably intrinsic (language-invariant) property of the equations, or of their space of solutions, rather than one that holds only relative to some further conventional choice of representational structure (e.g., certain coordinate systems).

In fact, the problems run deeper than this.  As \citet{Fletcher2020} and \citet{WeatherallDogmas} have both argued, in somewhat different contexts, such syntactic criteria are often ambiguously defined---a fact often acknowledged even by those who advocate for and use them.  At worst, they are hyperintensional, in the sense that they distinguish between cases that should be viewed as equivalent because the cases clearly represent the same relationships between the same physical quantities.  The source of this---again, see \citet{Fletcher2020} and \citet{WeatherallDogmas} for detailed examples---is that ``standard special relativistic form'' is not invariant under trivial syntactic manipulations of the expression of an equation, and apparently depends on various conventions about how to express equations and what features do and do not properly belong to such a form of an equation.

One way to attempt to resolve these issues is to focus only on coordinate systems adapted to the geometrical structures of spacetime \citep{Wallace}, so as to capture the idea that an ``equation'' ``depends'' (only) on certain structures, to certain orders---and that equations that do have just these dependencies can be expected to have solutions with certain properties. From this perspective, the trouble is that neither the notion of ``equation'' (or some equivalence class thereof) nor that of ``dependence'' has been made adequately precise in syntactic terms.  
But if the syntactic structures well adapted to geometrical structures were chosen just to reflect the features of the latter, then another way to resolve these issues might be to develop a theory of ``dependence'' that invokes only form-independent features of systems of equations.\footnote{Some work in this direction may be found in \citet{GerochPDE}.}  But we will not attempt that project here.  Instead, we turn to two other criteria for when a theory is LSR already available in the literature, both of which have a more semantic character, in that they look at properties of \emph{solutions} to a system of equations, rather than the manner in which the equations are expressed.

\section{The Sonego-Faraoni Criterion}\label{sec:Sonego}

Responding to objections very similar to the ones just expressed, \citet{sonego1993coupling} propose a different strategy for identifying when a system of equations is LSR.   
They suggest that, at least in the context of linear scalar field theories, to say a theory is LSR means that ``the physical properties of wave propagation---rather than the \emph{form} of the equation---should reduce locally to those valid in flat spacetime [\ldots, i.e.,] that the physical features of the solutions be locally the same in both cases'' (p. 1185, emphasis original).
In other words, it is not the form of the dynamical equations that matter for being LSR, but the local properties of the solutions to those equations.
They go on to suggest that a necessary and sufficient condition for this property to hold is that the ``fundamental solutions'', or Green's functions, of a matter theory in curved spacetime should have certain qualitative properties in common with those of the same theory in flat spacetime.
Those properties characterize how a point-like perturbation to a matter field propagates in an initial value problem.\footnote{
   They demand that, in the limit as the spacetime points of the arguments in the Green’s function approach one another, the Green's function for a general spacetime is that of Minkowski spacetime under the same limit. \label{fn:Green's function}
}
Because general solutions to the initial value problem for a linear matter theory may be constructed by a convolution of the fundamental solutions with an initial matter distribution, the local propagation properties for initial distributions are determined by those for the fundamental solutions \citep[p.~43]{di2015nonequivalence}.

To spell this proposal out in more detail, we will focus on the particular example that Sonego and Faraoni consider, which is a scalar field on a generic spacetime.  
Fix a spacetime $(M,g_{ab})$ and let $B$ be the trivial $\mathbb{R}$ bundle over $M$.  
We are interested in whether a (linear) differential operator $P[g_{ab}]$ on $M$ is LSR;\footnote{
    Here we are focusing attention on a class of linear field theories, and so we suppress the dependence on fields other than the metric since the space of solutions forms a vector space.} 
establishing this will then permit us to say whether a given scalar field theory is LSR on all spacetimes.

Recall, first, that given such a differential operator $P[g_{ab}]$ on $B$ and a point $q\in M$, a ``fundamental solution'' $G_q$ for $P$ is a (real-valued) distribution (or generalized function) with the property that
\[
P[g_{ab}]G_q = \delta_q
\]
where $\delta_q$ is the delta distribution based at the point $q$.\footnote{See \citet[Appendix A]{Geroch+Weatherall} for a compact introduction to tensor distributions, or \citet{distributionGR} for a longer treatment, though with somewhat different notational conventions.}  What this means is that given any smooth scalar field $\chi$ of compact support---a ``test field''---we have that 
\[
(P[g_{ab}]G_q)\{\chi\} = \delta_q\{\chi\} = \chi(q) 
\]
where by $D\{\chi\}$, for some distribution $D$ and test function $\chi$, we mean the action of the distribution on the test function.  

Consider, now, the class of differential operators given by, for some fixed number $m\in\mathbb{R}$
\[
P_{\xi}[g_{ab}] = g^{ab}\nabla_a\nabla_b + m^2 + \xi R
\]
where $\nabla$ is the Levi-Civita operator associated with $g_{ab}$, $R$ is the scalar curvature associated with $\nabla$, and $\xi$ is a real number. For any choice of $\xi$, we can take this expression to determine a scalar matter theory, by imposing, on the trivial $\mathbb{R}$-bundle over any spacetime $(M,g_{ab})$, a differential equation $P_{\xi}[g_{ab}]\varphi = 0$.  Call these theories ``P-$\xi$'' theories.  This class of theories is of interest because, for any $\xi\in \mathbb{R}$, the resulting theory coincides with Klein-Gordon theory (for appropriate $m$) on flat spacetime, but differs on other spacetimes. 

Sonego and Faraoni are interested in which P-$\xi$ theory, if any, should count as LSR.  Of course, one candidate answer stands out immediately: one might take the theory with $\xi=0$, i.e., the theory we have previously called Klein-Gordon theory, to be LSR, because it agrees syntactically with the Klein-Gordon equation on Minkowski spacetime.  But Sonego and Faraoni reject this answer, on the grounds that it violates the principle concerning the qualitative properties of solutions. The argument turns on the following facts.\footnote{See, e.g., \citet[\S 5.7]{friedlander1975} for a discussion of these facts in a detailed a rigorous setting.}

\begin{fact}
If $(M,g_{ab})$ is flat and $M$ is four-dimensional, then for any point $q\in M$, the fundamental solutions $G_q$ associated with $P_{\xi}[g_{ab}]$ are supported only on the light cone based at $q$.
\end{fact}

This property, of having fundamental solutions supported only on lightcones, is sometimes known as ``Huygens' Principle''.  Huygens' principal holds in flat spacetime for Klein-Gordon theory---and thus for any P-$\xi$ theory, since $R$ vanishes there.  But it does not hold in general for all values of $\xi$---including for $\xi=0$ on curved spacetimes.

\begin{fact}
For arbitrary spacetime $(M,g_{ab})$, the fundamental solutions $G_q$ of $P_{\xi}[g_{ab}]$ associated with points $q\in M$ may have support at any point $p$ timelike related to $q$.  In particular, for $\xi=0$, if $R\neq 0$ at a point $q$, then fundamental solutions $G_q$ have support at points timelike related to $q$.
\end{fact}

So Klein-Gordon theory---the theory that results if we take $\xi=0$---has qualitative properties that are not shared by the flat spacetime version of the theory.  The fundamental solutions have ``tails'' in the interior of the lightcone.   On the other hand, there is a (unique) value of $\xi$ for which $P_{\xi}[g_{ab}]$ does satisfy Huygens principle for all Lorentzian metrics $g_{ab}$ on four-manifolds $M$.

\begin{fact}For any spacetime $(M,g_{ab})$, with $M$ four dimensional, and any point $q\in M$, the fundamental solutions $G_q$ are supported only at points null related to $q$ if and only if $\xi=1/6$.
\end{fact}

Sonego and Faraoni go on to argue that for this reason, choosing $\xi=1/6$ is a natural and, indeed, preferable generalization of the Klein-Gordon equation to curved spacetime than the standard choice of $\xi=0$. Their full argument for this, and especially the consequences they draw from it, need not concern us here.  The key idea is that they use the LSR character of the resulting equation to distinguish P-$1/6$ theory from P-$\xi$ theory for other values of $\xi$, including 0.  In doing so, they suggest a precise necessary condition for certain differential operators to be LSR:

\begin{description}\item[LSR Chez Sonego \& Faraoni] A differential operator that satisfies Huygens' Principle in flat spacetime is LSR only if it satisfies Huygens' Principle in curved spacetime.\footnote{
    This is actually a special case of their necessary condition, which is described most succinctly in a later publication as the requirement ``that the local structure of the Green function [the fundamental solutions] associated with a physical law be the same in a curved and in a flat spacetime'' \citep[p.~43]{di2015nonequivalence}. 
    Although they do not give a complete account of what the ``local structure'' of the fundamental solutions amounts to, they do count satisfaction of Huygens' Principle as one example of it.
    See also footnote \ref{fn:Green's function} for more on what Sonego and Faraoni take to be constitutive of the fundamental solutions. 
}
\end{description}

Sonego and Faraoni's (partial) proposal for what it means for a differential operator to be LSR is an example of the sort of ``semantic'' alternative to the ``syntactic'' accounts discussed above that we suggested was needed, since it concerns the supports of certain solutions of an equation, rather than the expression of that equation.  
But it does have some surprising features, most notable of which is that it yields a non-standard answer to the question ``what scalar wave equations are LSR?''  
Now, Sonego and Faraoni take this to be an advantage of their proposal, because they think it provides guidance on how to generalize the Klein-Gordon equation to curved spacetime. 
But it leads to difficulties elsewhere.  
For instance, the source-free Maxwell equations also satisfy Huygens' principle in Minkowski spacetime (and in certain other contexts), but they do \emph{not} satisfy it in general in curved spacetimes \citep[Ch.~VIII]{HarteTails,belger1997survey,gunther1988huygens}.  
Thus, it would seem that Sonego and Faraoni are either committed to the claim that the source-free Maxwell equations are not LSR (or, in their own terms, do not satisfy the equivalence principle), or else that the equations must be modified in curved spacetime.  
To make matters worse, it is not clear that there exists \emph{any} modification to the equations that would make them LSR on the Sonego-Faraoni criterion, for any change that would re-establish Huygens' principle would seem to spoil the conformal invariance of the equations, a distinct local property of Maxwell wave propagation in Minkowski spacetime that guarantees that they are scale-free.
Introducing a spatial or temporal scale into the equations seems to lead to a scale for the development of the ``tail phenomena''---propagation of fundamental solution within, not just on, the lightcone---hence the failure of Huygens' Principle.

\section{The Ehlers-style Criteria}\label{sec:Ehlers}
\subsection{Two Versions: $g$-approximating and $\eta$-approximating}
Yet another semantic way---or, as we shall see, pair of ways---of thinking about what it means for a theory to be locally specially relativistic is essentially articulated by \citet{ehlers1973survey}.\footnote{It is worth noting that \citet[pp.~18--20]{ehlers1973survey} \emph{also} proposes a syntactic characterization of this property in his own statement of the SEP, which is closely related to Brown's statement.  Ehlers suggests that that version of the principle, along with results he proves along the lines of our Theorem 1 of Part 1, imply that what we are calling the Ehlers-style conditions will hold for theories that satisfy the syntactic criterion---though he does not prove this, and it is not clear that a proof would be possible.  We are extracting from his discussion there and elsewhere a pair of different, more semantic statements of what being LSR might mean.  The upshot will be a conflict between, on the one hand, each of those statements and, on the other hand, that of Sonego and Faraoni.}  He writes:
\begin{quote}\singlespacing Suppose a physical process which takes place in a small spacetime domain has been analyzed in the framework of special relativity theory, and one wishes to describe an analogous process occurring in a gravitational field without having a complete general-relativistic analysis---a common situation indeed. Then one may proceed as follows. One chooses an event $p$ in the general-relativistic spacetime model, situated where the process is assumed to happen, and identifies it tentatively with an event $p'$ of the special relativistic model. One identifies, moreover, the flat spacetime of the latter model with the tangent space $M_p$ of the curved spacetime $M$ at $p$, and one uses the exponential map $\exp_p$ to transfer all events of interest which form part of the special relativistic model of the process in question, from $M_p$ to $M$. The chronometrical relations between the events thus obtained will then, \textit{because of [results analogous to Theorem 1 of Part 1]}, differ arbitrarily little from those of the original, flat space theory, provided the gravitational tidal field $R^a_{~bcd}$, evaluated in a suitable orthonormal frame at $p$ adapted to the process, is sufficiently small; and one can estimate the deviations if one knows the curvature. It is therefore plausible to assume that [if a theory is LSR then], except for error of the order of the deviations just referred to, the model of the process obtained by applying the exponential map is a correct one. \citep[pp.~45--6]{ehlers1973survey}\end{quote}
The key idea we wish to extract from this passage is that a LSR matter theory in GR is one with solutions that, in general, locally approximate the solutions of the corresponding equations, \emph{mutatis mutandis}, in flat spacetime. We call this the (informal) \textit{$g$-approximating} version of the Ehlers criterion, for it posits, for any GR spacetime and any local special relativistic matter theory solution, the existence of a local GR solution that approximates it in the aforementioned GR spacetime. 

This idea is promising---and for theories for which it is true, it has clear pragmatic value, for it suggests how to model local matter dynamics approximately in GR using knowledge from SR.  But trying to articulate this criterion precisely and demonstrate whether candidate theories satisfy it turns out to be a subtle business.  (Ehlers does not attempt it.)  At issue are two things.  First, Ehlers suggests that the approximating curved spacetime solution should be ``correct'', presumably in the sense that it provides an adequate representation of a process analogous to the one that was originally modeled in Minkowski spacetime.  What does this notion of ``correctness'' amount to?  How do we identify ``situations'' across different spatiotemporal contexts? The second subtlety is even more basic.  What does it mean for one solution to ``approximate'' a correct one?  The answers to both questions are surely context-dependent, and not something amenable to an abstract, informative, once-and-for-all characterization.

Nonetheless, we suggest a schema for this $g$-approximating Ehlers-type semantic criterion of when a theory is LSR that affords enough flexibility to capture many relevant contexts in which ``being LSR'' is invoked.\footnote{We don't know of any limitations per se, except for the contrasting case of the $\eta$-approximating version of the criterion discussed later in this subsection.}  To get at the idea, fix some relativistic spacetime $(M,g_{ab})$, and choose any point $p\in M$.  Theorem 1 of Part 1 ensures that there exist, on sufficiently small neighborhoods of $p$, flat metrics $\bar{g}_{ab}$ with the properties that $\bar{g}_{ab} = g_{ab}$ at $p$, and moreover, $\bar{\nabla}=\nabla$ at $p$, where $\bar{\nabla}$ is the Levi-Civita derivative operator associated with $\bar{g}_{ab}$.  Fix some such flat metric $\bar{g}_{ab}$.  We will say that a differential equation on $M$ is LSR if given any solution to the ``flat spacetime equation'', i.e., the equation that the theory assigns to $(U,\bar{g}_{ab})$, there exists a solution to the equation the theory assigns to curved spacetime---that is, $(U,g_{ab})$---that approximates it relative to some norm.  More precisely:\footnote{Throughout the definitions in this section, we implicitly assume a linear or affine structure on the fibers of the bundle, such that a norm of the form described can be defined.  This is not a significant restriction in practice, but it deserves to be noted.}

\begin{description}\item[LSR schema, Chez Ehlers, $g$-approximating version] Suppose we have a matter theory in the sense given above.  Then given any spacetime $(M,g_{ab})$, with associated field bundle $B\rightarrow M$ and differential operator $P[g_{ab},\varphi]$, and given any point $p\in M$, any flat metric $\bar{g}_{ab}$ near $p$ agreeing, to first order, with $g_{ab}$ on $p$, any solution $\bar{\varphi}:O\rightarrow B$ to $P[\bar{g}_{ab},\bar{\varphi}]\bar{\varphi}=0$ near $p$, and any $\epsilon > 0$, there exist, on a sub-neighborhood $\hat{O}\subseteq O$ of $p$, a field $\varphi:\hat{O}\rightarrow B$ such that $P[g_{ab},\varphi]\varphi=0$ and a sub-sub-neighborhood $U\subseteq \hat{O}$ of $p$ on which $||\varphi-\bar{\varphi}||<\epsilon$, where $||\cdot||$ is a suitable, context-dependent norm. \end{description}

Here the differential operator $P[\bar{g}_{ab},\varphi]$ is to be interpreted as the operator in the ``flat spacetime'' version of the equation---that is, the equation that the matter theory would impose on the spacetime $(O,\bar{g}_{ab})$\footnote{As in Part 1, we will not explicitly indicate that a $g_{ab}$ defined on a manifold is restricted to some open subset of that manifold.}---and $P[g_{ab},\varphi]$ is the operator associated with the curved spacetime equation.  (By our definition of matter theory in section \ref{sec:preliminaries}, there can be no ambiguity about what those operators should be.)  Note that while we work with the spacetime $(U,\bar{g}_{ab})$, this spacetime can always be isometrically embedded into Minkowski spacetime---and thus, solutions to the ``special relativistic'' equations can be pulled back to $(U,\bar{g}_{ab})$ to determine the sort of local flat spacetime solutions under consideration.  In that way, we can think of the condition as saying: a theory is LSR if for any (global or local) solution to the ``special relativistic'' version of the equation, there exists a GR solution that approximates it near any point of any spacetime.

There is another interpretation of Ehlers' remarks available, especially in the context of other passages, where he affirms that ``a theory of spacetime and gravitation should be such that [among other conditions] it agrees locally with special relativity; \ldots [this is, in other words,] the requirement that special relativity be valid to a good approximation locally everywhere even in the presence of gravitation'' \citep[pp.~19--20]{ehlers1973survey}.  
According to this (informal) \textit{$\eta$-approximating} version of the Ehlers criterion  one should take a theory to be LSR if for every solution in curved spacetime, and every point, there exists a flat spacetime solution that approximates it near that point.  This direction might be taken to capture the idea that every local material phenomenon that occurs in curved spacetime can be suitably approximated by an analogous situation in flat spacetime.  
Note the difference in quanitifier domains here: in the $g$-approximating version, the criterion quantifies over all local flat spacetime solutions and all local regions of an arbitrary curved spacetime, asserting the existence in the latter of a curved spacetime solution that approximates the given flat spacetime solution; in the $\eta$-approximating version, the criterion quantifies over all local curved spacetime solutions and all local regions of flat spacetime, asserting the existence in the latter of a flat spacetime solution that approximates the given curved spacetime solution.
The more formal version runs thus:

\begin{description}\item[LSR schema, Chez Ehlers, $\eta$-approximating version] Suppose we have a matter theory still in the sense given above.  Then given any spacetime $(M,g_{ab})$, with associated field bundle $B\rightarrow M$ and differential operator $P[g_{ab},\varphi]$, and given any point $p\in M$, any flat metric $\bar{g}_{ab}$ near $p$ agreeing, to first order, with $g_{ab}$ on $p$, any solution $\varphi:O\rightarrow B$ to $P[g_{ab},\varphi]\varphi=0$ near $p$, and any $\epsilon > 0$, there exist, on a sub-neighborhood $\hat{O}\subseteq O$ of $p$, a field $\bar{\varphi}:\hat{O}\rightarrow B$ such that $P[\bar{g}_{ab},\bar{\varphi}]\bar{\varphi}=0$ and a sub-sub-neighborhood $U\subseteq \hat{O}$ of $p$ on which $||\varphi-\bar{\varphi}||<\epsilon$, where $||\cdot||$ is a suitable, context-dependent norm. \end{description}

We do not adopt a position on which version of the Ehlers schema is superior, nor on which is a more accurate statement of Ehlers' own views.  Both are precise criteria that might be appropriate to consider under various circumstances. 
For instance: the $g$-approximating version captures a statement of a sense in which GR matter theories might be said to subsume, at least locally and in approximation, SR matter theories; and the $\eta$-approximating version captures a statement of a sense in which, at least locally, solutions to SR matter theories approximate solutions to the corresponding GR matter theories.
But note the difference, in either version, from the Sonego and Faraoni condition:  Whereas they consider perfect agreement with respect to a specific qualitative property of \emph{fundamental} solutions---those arising from a delta distribution source---the Ehlers schemas consider approximate agreement of \emph{arbitrary} solutions.

\subsection{The Contextuality of Appropriate Norms}
Both version of the Ehlers schema invoke some norm $||\cdot ||$ on local sections of $B$.  Indeed, this is why we refer to each of them as a ``schema''.  One finds different criteria by choosing different norms.  What should this norm be? Ehlers does not say.  (He even suggests that the right norm will be whichever one makes the condition true of his favorite matter theories!)  We will not propose one here; indeed, we take the freedom to choose an appropriate norm to suit the situation at hand to be a reflection of the context-sensitivity associated with ``correct'' and ``approximate'' solutions in the passage from Ehlers quoted above.\footnote{Compare these remarks with those of \citet{fletcher2016similarity}, who argues that criteria of similarity and approximation have this context-dependency quite generally.}

That said, we will note some minimal desiderata that seem reasonable to expect in all cases.  The first is that whatever norm one adopts, it should measure not only the values of $\varphi$, but also their derivatives up to at least the order of the differential operators associated with a theory.  Because in general the values of the derivatives of a field at a point to some particular order are independent of those values to lower (including zero) order, this ensures that all (rather than a strict subset) of the data that serve as input to the equation associated with the differential operator enter into our notion of approximation.  Thus, the norm may depend on the \textit{equations} whose solutions are locally approximated---and it may even depend on how an equation is characterized, since generally, any partial differential equation of order greater than 1 can be reformulated as an equation of order 1 by introducing auxiliary fields and further equations to guarantee consistency.  (This dependence of the norm on the formulation of an equation might seem odd at first glance, but it is not so troubling once we recognize that when an order $k$ equation is cast in first-order form, one represents its higher-order derivatives as new fields subject to first order equations. So, a norm sensitive to order $k$ for the original field equation is mathematically equivalent to a norm sensitive to order 1 for the new first-order field equation, since the new equation simply replaces the higher-order derivatives with new fields.)

A second consideration is that as long as one keeps to the Ehlers-like schemas as we have stated them, whatever norm is chosen must be such that the difference in norm between two fields does not approach zero as one considers smaller and smaller neighborhoods of the point $p$, even when the difference in field values remains large.  In other words, the norm should be \emph{fiber-wise separating}: when restricted to a single point, the norm of the difference between two values should vanish only if those values are identical.  This rules out candidates such as Sobolev norms that involve integrating fields over some neighborhood, since those integrals will shrink to $0$ as the neighborhoods shrink, even if the differences in the fields remain bounded away from $0$; on the other hand, modified Sobolev norms that take into account the size of the neighborhood may be appropriate in some contexts.  

The ambiguity regarding what norm should be chosen makes it difficult to evaluate unequivocally whether actual matter theories turn out to be LSR according to an Ehlers-style criterion.  But as we will presently show, for at least some norms that meet the desiderata already laid out, standard theories, including Maxwell's theory and the Klein-Gordon theory, \emph{are} LSR by the resulting Ehlers criteria.  For this purpose, we consider $C^k$ supremum norms on the neighborhoods $U$ of $p$ considered in the criterion schemas, where $k$ is the order of the system of equations under considerations and the supremum norm is defined relative to some positive definite metric, $h_{ab}$.\footnote{As we shall see, for some purposes, it is convenient to take $k=1$ and work only with first-order formulations of systems of equations.  Note that if we were to adopt $k=0$ for our norm, any equation with the property that, given any field value $x\in\pi^{-1}[p]$, it has a solution through $x$, would automatically be LSR.  But this would include, for instance, \textit{all} first-order quasilinear systems.}  (These are the norms determined by the distance functions $d_U(f,f';h,k)$ defined in Part 1.)  They meet the desiderata already sketched in that they measure not just field values, but also derivatives, and they are fiber-wise separating.  But we do not claim that these are somehow a canonical choice---just one that is convenient for present purposes and meets the minimal conditions we have set out.

Consider, first, vacuum Maxwell theory and the $g$-approximating version of the Ehlers LSR criterion with the $C^1$ supremum norm.  
\begin{prop}
Fix any spacetime $(M,g_{ab})$ and choose any point $p\in M$.  Let $\bar{g}_{ab}$ be a flat metric on a neighborhood $O\subseteq M$ of $p$, which agrees, to first order, with $g_{ab}$ at $p$.  Finally, let $\bar{F}_{ab}:O\rightarrow B$ be a solution to Maxwell's equations on flat spacetime, 
\begin{subequations}\label{eq:MaxFlat}
\begin{align}
    \bar{g}^{ab}\bar{\nabla}_a \bar{F}_{bc} &= \mathbf{0},\\
    \bar{\nabla}_{[a}\bar{F}_{bc]} &= \mathbf{0}.
\end{align}
\end{subequations}
Given any positive definite metric $h_{ab}$ and any $\epsilon>0$, there exists a solution $F_{ab}$ to Maxwell's equations on $(O,g_{ab})$,
\begin{subequations}\label{eq:MaxCurved}
\begin{align}
    g^{ab}\nabla_a F_{bc} &= \mathbf{0},\\
    \nabla_{[a}F_{bc]} &= \mathbf{0},
\end{align}
\end{subequations}
such that there exists a neighborhood $U\subset O$ of $p$ on which $||F_{ab}-\bar{F}_{ab}|| = d_U(F,\bar{F};h,1) < \epsilon$, where $||\cdot||$ denotes the $C^1$ supremum norm.\footnote{In fact, we show something stronger: we show that there exists a field $F_{ab}$ that solves Maxwell's equations on $(O,g_{ab})$ and which is such that, for any positive definite metric $h_{ab}$ and $\epsilon>0$, there exists $U \subset O$ such that $d_U(F,\bar{F};h,1) < \epsilon$.  That is, a single curved spacetime solution suffices for all positive definite metrics if the neighborhood $U$ is allowed to vary accordingly.}  
\end{prop}
\begin{proof}
It suffices to show that there exists a solution $F_{ab}$ to Eqs. \eqref{eq:MaxCurved} such that $F_{ab}=\bar{F}_{ab}$ and $\nabla_a F_{bc} = \nabla_a \bar{F}_{bc}$ at $p$.

We do this as follows.  Choose an open set $\hat{O}\subseteq O$ containing $p$.  Without loss of generality, we suppose that $\hat{O}$ is the (interior of the) domain of dependence of some achronal, spacelike hypersurface $S$, where $S$ contains $p$.  Define a new field $F_{ab}$ on $O$ as:
\[F_{ab} = \bar{F}_{ab} + \hat{F}_{ab},\]
where $\hat{F}_{ab}$ is the unique solution to Maxwell's equations on $(\hat{O},g_{ab})$ that vanishes on $S$ and whose sources are given by $-\nabla_a \bar{F}^{ab}$, where indices are raised with $g_{ab}$.  (If no such solution exists, say because $\hat{F}_{ab}$ becomes singular, choose $\hat{O}$ to be the maximal neighborhood of $S$ on which a unique solution does exist; such a neighborhood is guaranteed to exist by the standard existence and uniqueness theorems governing symmetric hyperbolic first-order linear PDEs.) Then $F_{ab}$ is a solution to Eqs. \eqref{eq:MaxCurved} on $\hat{O}$.  Moreover, $F_{ab}$ agrees with $\bar{F}_{ab}$ at $p$, since $\hat{F}_{ab}=\mathbf{0}$ there.  

It remains to show that $\nabla_aF_{bc}=\nabla_a\bar{F}_{bc}$ at $p$.  And for this it suffices to show that $\nabla_a \hat{F}_{bc} = \mathbf{0}$ at $p$.  To do so, we note that by the symmetric hyperbolic character of Maxwell's equations, there exists a smooth tensor field $\chi^{amnxy}$ on $M$ with the following properties: (1) given any smooth causal covector $n_a$, there exists a tensor field $\zeta_{mnxy}$ such that $n_a\chi^{amnxy}\zeta_{xyst}=\delta^{[m}{}_{s}\delta^{n]}{}_{t}$; and (2) if $F_{ab}$ solves Eqs.~\eqref{eq:MaxCurved}, then $\chi^{amnxy}\nabla_{a}F_{mn} = 0$ \citep{GerochPDE}.  Now let $n_a$ be the (timelike) covector field normal to $S$, and let $q^a{}_b = \delta^a{}_b - n^an_b$ be the projection orthogonal to $n^a$ (i.e., the projection onto the tangent space of $S$) at each point of $S$.  Now, since the sources for $\hat{F}_{ab}$ vanish at $p$ (because $g_{ab}$ and $\bar{g}_{ab}$ agree there to first order), at $p$ we have that $\hat{F}_{ab}$ satisfies: $\mathbf{0}=\chi^{bmnxy}(n^an_b + q^a{}_b)\nabla_a \hat{F}_{mn}$, which implies that $(n_b\chi^{bmnxy})n^a\nabla_a\hat{F}_{mn} = -q^a{}_b\chi^{bmnxy}\nabla_a \hat{F}_{mn} = \mathbf{0}$, where the final equality holds due to the fact that $\hat{F}_{ab}$ is constant ($=0$) on $S$, and thus its derivatives tangent to $S$ must vanish.  Recalling that $n_b\chi^{bmnxy}$ is invertible, we conclude that $n^a\nabla_a\hat{F}_{mn}=\mathbf{0}$.  It follows that $\nabla_a \hat{F}_{mn} = (n^m_a + q^m{}_a)\nabla_m \hat{F}_{bc}=\mathbf{0}$ at $p$, and thus $\nabla_a F_{bc} = \nabla_a \bar{F}_{bc}$ there.
\end{proof}

So source-free Maxwell theory is LSR according to the $g$-approximating Ehlers-style criterion with the $C^1$ supremum norm.  An identical argument applies if we include sources in Maxwell's theory, so long as we understand the ``theory'' involving those sources to be such that: (a) if $g_{ab}=\bar{g}_{ab}$ at a point $p$ to first order, then the source terms associated with the sourced Maxwell's equations in curved and flat spacetime also agree at $p$; and (b) whatever matter fields give rise to the source term themselves satisfy some first-order quasilinear symmetric hyperbolic equations.  (This obtains, for instance, for source terms derived from standard charged matter fields, such as complex Klein-Gordon fields or Dirac fields.)  And the same argument applies to Klein-Gordon theory, as well; this is particularly clear if we write the theory in first-order form by introducing an auxiliary field $\psi_a$ and imposing additional equations to require that $\psi_a=\nabla_a\varphi$, since in this form one can invoke the symmetric hyperbolic character of the equations in the same ways as for the Maxwell case.  

With little modification, the foregoing arguments show that both Maxwell and Klein-Gordon theories are LSR according to the $\eta$-approximating Ehlers-style criterion, too. In fact, there are also other, independent arguments available to establish this.  For instance, in the Klein-Gordon case, one can approximate any curved spacetime solution, to second order at a point, by an appropriately chosen sum of plane wave solutions in flat spacetime (plus a correction term).  This is so even if the qualitative behavior of the curved spacetime solution in a neighborhood of the point is not well represented by a sum of plane waves.\footnote{This observation highlights how the $C^2$ supremum norm, though it meets our desiderata for non-triviality, may still be too weak as a standard of approximation for some purposes that demand identity of such qualitative behavior \citep[cf.][]{fletcher2016similarity}. } 

\subsection{Further Versions}
Another well-motivated Ehlers-style criterion is the conjunction of the $g$-approximating and $\eta$-approximating versions. If one does adopt this \emph{both-ways Ehlers criterion}, then one can show an analogue of Theorem 5---the universal approximation theorem---from Part 1.
\begin{prop}
If a theory assigning equation $P[g_{ab},\varphi]=0$ is LSR according to the both-ways Ehlers criterion, then given any two spacetimes $(M,g_{ab})$ and $(M',g'_{ab})$, any points $p\in M$ and $p'\in M'$, a map $\chi:O\rightarrow O'$ between neighborhoods of $p$ and $p'$ such that $\chi^*(g'_{ab})=g_{ab}$ at $p$ (and likewise for their derivatives) as guaranteed exists by Theorem 5 of Part 1, and any solution $\varphi$ to the equation $P[g_{ab},\varphi]$ in a small neighborhood of $p$, there exists a solution $\varphi'$ to the equation $P[\chi^*(g_{ab}),\varphi']=0$ that approximates $\varphi$ arbitrarily well near $p$, relative to the same norm that realizes the both-ways Ehlers criterion.  
\end{prop}
\noindent That is, every matter theory that is ``locally specially relativistic'' in this sense is also, equally well, ``locally de Sitter'', ``locally Schwarzschildian'', etc.

Another variation on the Ehlers-type criteria as we have stated them is to demand more from an ``approximate solution'' than is captured by the norms we have considered here.  For instance, following Ehlers, we have stated the schema in terms of approximation in a neighborhood of a point.  But we know from Theorem 1 of Part 1 that relativistic spacetimes are approximately locally flat in a stronger sense: there exist flat metrics that agree with a given curved metric, to first order, along a curve.  Could one formulate a version of the Ehlers condition that involves approximation along a curve?  

There is a sense in which the answer is clearly ``yes'': one could require, for instance, that for any $\epsilon>0$ there is a solution to the curved-spacetime equation that approximates any given flat spacetime solution in sufficiently small neighborhoods of sufficiently small segments of any curve.  But nothing is added by this approach, since any neighborhood of a point $p$ will necessarily contain some small segment of any curve through $p$.  

One might consider, instead, fixing a segment of a given curve, and then requiring arbitrary good approximation, relative to some norm, within any sufficiently small neighborhood of that fixed segment.  We do not claim that this is a fruitless idea, but it does run into difficulties very quickly.  In particular, for the $C^k$ supremum norms we have considered, a curved-spacetime solution will approximate a given flat-spacetime solution arbitrarily well in small neighborhoods of a fixed segment of a curve only if the two solutions agree to $k$th order at every point of the segment.  But it is hard to see how this could be achieved for hyperbolic systems, such as Maxwell's equations, since one should expect small deviations between the solutions off of $\gamma$ to propagate to and intersect $\gamma$, leading to disagreements between the solutions on $\gamma$ arbitrarily close to any point of agreement.

Another approach, related to the $g$-approximating criterion, would be to require not only that there exist solutions in curved spacetime that approximate flat spacetime solutions in a neighborhood of a point, but also that the deviations between those solutions are bounded by some function that is ``small'' near the curve.  We will not pursue this approach in detail, but we do note that the approximating solutions we constructed for Maxwell's equations above have this character.  In particular, using Theorem 1 of Part 1, one can choose $\bar{g}_{ab}$ and $\hat{F}_{ab}$ so that the source term for $\hat{F}_{ab}$, $\nabla_a \bar{F}{}^{ab}$, vanishes along a segment of a causal curve $\gamma$ passing through $p$.\footnote{For further details on this approach, see the arguments in \citet[\S 4]{Geroch+Weatherall}.}  Meanwhile, one can invoke well-known energy inequalities in the context of hyperbolic systems to show that the ``total size'' of $\hat{F}_{ab}$, say in the $L^2$ norm on a given causal diamond, is bounded by the integral of the ``total size'' of this source term, which can be kept small on suitably chosen regions because it vanishes on $\gamma$.\footnote{See, for instance, the inequality derived in Appendix B of \citet{GerochPDE}.  Note here that $\hat{F}_{ab}$ coincides with the $0$ solution on a surface $S$, that both exactly satisfy their respective equations, and that, using his notation, $k=k'$.}  There results a family of inequalities of the form:
\[
||\hat{F}||_V \leq \int_V X dV,
\]
where $X$ is a function vanishing on a segment of $\gamma$ and depending on $\bar{F}_{ab}$ and the curvature of $g_{ab}$, and $V$ is an appropriately chosen causal diamond.  We do not claim that this is the most perspicuous or strongest such inequality that one could derive; we note it only as a suggestion for a further direction one might pursue to capture a somewhat stronger sense in which a theory might be LSR.

\section{Comparing the Foregoing Criteria}\label{sec:comparisons}

We will now turn to comparing the three (families of) criteria we have just discussed.  First, we consider the relationship between the Ehlers-style criteria of section \ref{sec:Ehlers} and the Sonego-Faraoni one of section \ref{sec:Sonego}.  As we noted previously, one concern about the Sonego-Faraoni approach is that it yields odd verdicts on cases that might have seemed like canonical examples of LSR theories, including both Klein-Gordon theory and Maxwell theory.  On the other hand, we have seen that with appropriate choices of norm, Maxwell and Klein-Gordon theories are both LSR according to the Ehlers criteria.  These arguments reveal a key difference between the Ehlers-type criteria we have discussed and the Sonego-Faraoni criterion.  They show that being LSR on the Ehlers-type criteria, at least with the class of norms considered here, is not sufficient to be LSR on the Sonego-Faraoni one.  

Of course, this observation leaves open the possibility that the Sonego-Faraoni condition is strictly stronger.  That is: is being LSR on the Sonego-Faraoni criterion sufficient for being LSR on the Ehlers-type criteria?  The answer is ``no''.
\begin{prop}
There exist matter theories that are LSR on the Sonego-Faraoni criterion but are not LSR on the $g-$ or $\eta-$approximating Ehlers-type criteria with a $C^k$ surpremum norm, where $k$ is the order of the differential operator for the matter theory..
\end{prop}
\begin{proof}
To find an example that witnesses this proposition, consider again the P-$\xi$ theory as presented by Sonego and Faraoni.  Choose an arbitrary spacetime $(M,g_{ab})$ with nonvanishing scalar curvature, and let $p$ be a point at which $R\neq 0$.  Now choose a neighborhood $U$ of $p$ admitting a flat metric $\bar{g}_{ab}$ agreeing, to first order and at $p$, with $g_{ab}$; and let $\bar{\varphi}$ solve the flat spacetime equation $P_{1/6}[\bar{g}_{ab}]\bar{\varphi} = \bar{g}{}^{ab}\bar{\nabla}_a\bar{\nabla}_b\bar{\varphi} + m^2\bar{\varphi} = 0$.  We assume, without loss of generality, that $\bar{\varphi}$ has been chosen to be non-vanishing at $p$.  Now suppose, for contradiction, that there exists a solution $\varphi$ to the curved spacetime equation, $P_{1/6}[g_{ab}]\varphi = g^{ab}\nabla_a\nabla_b\varphi + m^2\varphi + R\varphi/6 = 0$ such that, for some positive definite metric $h_{ab}$ defined near $p$ and any $\epsilon>0$, there exists a neighborhood $U$ on which $d_{U}(\varphi,\bar{\varphi};h,2) < \epsilon$.  It would follow that at $p$, $\varphi=\bar{\varphi}$, $\nabla_a\varphi =\nabla_a\bar{\varphi}$, and $\nabla_a\nabla_b\varphi = \nabla_a\nabla_b\bar{\varphi}$, since otherwise, the norm, relative to $h_{ab}$, of the difference between any of these quantities would bound the supremum of those differences away from $0$, as it is fiberwise separating.  But then by construction, we have, at $p$, $g{}^{ab}\nabla_a\nabla_b\varphi + m^2\varphi=g{}^{ab}\nabla_a\nabla_b\bar{\varphi} + m^2\bar{\varphi} = 0 \neq -\xi R\varphi$.  And so any field $\varphi$ that agrees to second order with $\bar{\varphi}$ at $p$ fails to satisfy the curved spacetime equation at $p$.  Of course, this result holds, in particular, for $\xi=1/6$.\footnote{Observe that the same argument would go through, \emph{mutatis mutandis}, if we considered a first order formulation of P-$\xi$ theory and a $C^1$ norm.}  
\end{proof}
This result depended on a particular choice of norm, as did the results that established that Maxwell and Klein-Gordon theories were LSR on the Ehlers criteria.  We take the norm that we have adopted to be a weak one that meets the minimal desiderata noted above.  If one adopted a \emph{stronger} norm, the result that the Sonego-Faraoni criterion does not imply the Ehlers-type criterion would still hold; likewise, with a weaker norm, the fact that Ehlers-type criteria do not imply the Sonego-Faraoni criterion would still hold.

So these criteria come apart---and indeed, are incompatible, in the sense that each is satisfied, for the P-$\xi$ theory, for unique, but distinct, values of $\xi$, at least for our choice of norm.  It follows that different authors who endorse the claim that matter theories in general relativity should be LSR not only differ in how to state that requirement precisely, they also substantially disagree about what theories should be taken to meet the requirement in the first place.  

What about the syntactic criterion?  As we have already argued, it is not sufficiently clear what the criterion is supposed to be to attempt to prove anything formally about its relationship to the other two criteria. If the criterion is hyperintensional, as we suggested, then we would expect that it neither implies nor is implied by either of the semantic criteria---either those of Sonego-Faraoni or of Ehlers. What advocates of the syntactic criterion have said about particular examples seems to bolster this case.
For instance, \citet[p.~171]{Brown2005} suggests that \emph{both} what we have called Klein-Gordon theory (in section \ref{sec:Sonego}) and P-$1/6$ theory are LSR, suggesting that the syntactic criterion is weaker than all of the others.\footnote{Brown writes that the P-$1/6$ theory is ``locally Lorentz covariant'', which we take to imply that it is LSR, especially in light of the later arguments by \citet{ReadBrownLehmkuhl2018}; but he also indicates that $\xi=1/6$ is ``permissible'' to recover the ``local physical properties of wave propagation'' in special relativity, but \emph{not} to recover the ``form of the [Klein-Gordon] equation'' \citep[p.~172]{Brown2005}.  So we are not certain how the syntactic criterion rules on this case.}  On the other hand, \citet[p.~171]{Brown2005} identifies Jacobsen-Mattingly theory \citep{jacobson2001gravity} as a theory that he considers to be not LSR (and \citet{ReadBrownLehmkuhl2018} concur), but we think it is plausible that the Jacobsen-Mattingly theory \emph{does} satisfy the Ehlers criteria for reasonable norms.\footnote{Note that, without matter, the Jacobsen-Mattingly theory is equivalent to a sector of Maxwell's theory coupled to a charged dust \citep[024028-1]{jacobson2001gravity}.  Thus the arguments given above that Maxwell's theory satisfies the Ehlers criteria \emph{almost} establish that Jacobsen-Mattingly theory does as well---``almost'' because the dust equations are not symmetric hyperbolic \citep{GerochPDE}, and so a different method is needed to establish the result.  That said, a modified Jacobsen-Mattingly theory that included a pressure term presumably \emph{would} satisfy the Ehlers criteria, by the arguments above, so long as the effective equation of state satisfied certain further criteria.  But a detailed assessment of the Jacobsen-Mattingly theory is beyond the scope of this paper.  We note it only because Brown identifies it as an example of a theory that is not LSR, but it is not at all obvious that it violates the Ehlers criteria.}  Suffice it to say that work is needed to identify precise mathematical statements of the syntactic criteria to truly settle these issues, and so we leave them aside for now.

\section{The Second Ambiguity, or Where Are We Now?}\label{sec:whoCares}

At the beginning of section \ref{sec:introduction}, we noted two ambiguities related to the ``matter dynamics'' interpretation of local flatness.  Our discussion thus far has focused on the first ambiguity: what does it mean to say that a matter theory is LSR?  But we have not yet touched on the second ambiguity, which was whether the LSR character of matter theories is a convenient property that happens to obtain of familiar theories, or a constitutive condition, the violation of which would mean that a theory fails to be adequate (or, adequately ``relativistic''), i.e., satisfy the interpretive postulates of relativity theory.  Is a matter theory in curved spacetime somehow unacceptable if it is not LSR?  

\citet[p.~151]{Brown2005} famously remarked that the ``chronogeometric significance of the [metric] field is not an intrinsic feature of gravitational dynamics, but earns its spurs by way of the strong equivalence principle,'' which is Brown's expression of what it means for a matter field to be LSR. (See also \citet[\S4.1]{ReadBrownLehmkuhl2018}.) Yet it is not clear to us that the spacetime metric cannot be interpreted in the usual way, nor that anything else \emph{does} go wrong if a theory is not LSR. For instance, we do not find anything inherently problematic about either Klein-Gordon theory or the alternative advocated for by Sonego and Faraoni. And although we have seen that Klein-Gordon and Maxwell theories are LSR in the sense of Ehlers for some norms, we have not shown that other realistic theories, such as non-Abelian Yang-Mills theories, are LSR on those criteria.  Indeed, since the methods we used for Maxwell and Klein-Gordon theories made apparently essential use of the linear character of those equations, it is not immediately clear how to extend those arguments to non-linear theories.  So if failing to be LSR has problematic consequences, it remains open whether the theories that seem to in fact describe fundamental physics must face those consequences. 
We also do not have a definitive answer to what consequences there might be, at least in any detail. A fuller answer would show where these various LSR criteria fit within the network of interpretive principles, standard assumptions, and theorems of GR.\footnote{
    See  \citet{WeatherallPuzzleball} for a discussion of the background picture here (which Weatherall calls the ``Puzzleball conjecture'').  Briefly, the idea is that various standard assumptions and core principles, such as the geodesic principles, the conservation condition, Einstein's equation, etc.~are often mutually inter-derivable.  Removing or modifying one of these principles thus has downstream consequences for the theory that force (perhaps unexpected) interpretation changes.  \citet{WeatherallDogmas} makes some closely related comments about the (obscure) deductive consequences of the strong equivalence principle.}  

Another possible problem that is sometimes raised for theories that fail to be LSR is that they cannot explain, or are perhaps even incompatible with, the success of special relativity.  A full treatment of this issue would take us far afield, and so we largely set it aside here except to venture a preliminary remark: this explanatory task seems to be most closely related to what might be guaranteed by the $\eta$-approximating Ehlers-type schema (rather than the $g$-approximating one). Recall that, according to it, a matter theory is LSR if for every local region of any spacetime and any field solution thereon to the matter theory's equations of motion, there exists a local flat-spacetime solution that approximates it.  
This could fail, and yet special relativity might nonetheless be effective under circumstances where, for instance, curvature is sufficiently small, or where the actual field configurations encountered in nature happen to be similar to flat spacetime solutions, even though other solutions are not.
In a word, it may suffice for matter theories to be LSR only in the circumstances where their special relativistic counterparts have been successfully applied.

\section{Conclusion and Directions for Future Research}
\label{sec:conclusion}

In this Part of our two-part essay, we considered the ``matter dynamics'' interpretation of the idea that special relativity is locally valid in general relativity, which was the claim that, in some sense, matter in general relativity locally behaves as if it were in flat spacetime, at least approximately.  We isolated and assessed three criteria for when a matter theory should be said to be locally special relativistic, and argued that, insofar as they can be made sufficiently precise as to render clear verdicts, they are incompatible.  It apparently follows that there is no ``obvious'', or universally accepted, precise statement of this property lurking in the background of informal discussions.  It seems that anyone who wishes to attribute foundational significance to the LSR character of matter dynamics needs to provide further arguments for why their preferred version has the significance they wish to attribute to it.  In our view, the Ehlers-style criteria we have discussed are the most promising to develop further---or at least, they seem to us like the most natural way of capturing the claim that a matter theory is (approximately) LSR.  But we will not defend that claim further here.

We also briefly discussed what depends on matter being LSR.  We had little to say, here, because it remains unclear how the various LSR criteria depend on, or imply, other principles of GR.  But since many of the criteria under consideration turn out to be mutually incompatible, and there are examples of theories that are LSR on one criterion but not others (and vice versa) that have received serious attention within the physics community, it is hard to see firm grounds for insisting that non-LSR theories are illegitimate or incoherent, incompatible with GR as standardly understood, or more likely to be false (all else being equal).  Still, this question deserves further attention.  

We close our present investigation with comments on two further directions for future research.  The first picks up the theme from the end of the penultimate section.  As noted in Section 1 of Part 1, a principal motivation for many authors who have argued that general relativity is locally similar to special relativity has been to account for the success of the latter theory.  From that perspective, the approximate local flatness of relativistic spacetimes, or the LSR character of various matter theories, plays a certain role in explanations and/or in a reduction of one theory to another.  While we have focused on the status of (various versions of) the background claim that general relativity is, somehow, locally like special relativity, we have only superficially addressed the role that the relationships we have discussed might play in explanations of the success of special relativity.  In our view, it would be valuable to explore this issue more completely in light of the results and discussion of the present paper.  In particular: as suggested above, is it truly necessary that matter theories be LSR in order to account for the success of special relativity?  To what extent does the answer to that question depend on one's theory of intertheoretic relations?

The second direction for future research concerns an expression that we have not used here, but which is very often paired with discussions of the SEP and the LSR character of matter theories---viz. \textit{minimal coupling.}
According to the most common definition, a matter field is said to be minimally coupled to the spacetime metric, or minimally gravitationally coupled, when its equations,  ``written in abstract geometric form, differ in no way whatsoever between curved spacetime and flat spacetime; \dots [and when] written in component form, change on passage from flat spacetime to curved spacetime by a mere replacement'' of the Minkowski metric and compatible coordinate derivatives with a general metric and compatible covariant derivatives, respectively \citep[p.~387]{MTW1973}.\footnote{
    For readability, we have removed italics from the original quotation, which refers to the replacement as that of ``commas by semicolons'' in reference to the common respective notation for the coordinate and covariant derivatives.
    At variance with most other authors, \citet[p.~386]{MTW1973} actually refer to minimal coupling as ``Einstein's equivalence principle.''
}
Many authors take minimal coupling to be a \textit{sufficient} condition for the strong equivalence principle \citep[e.g.,][p.~171]{Brown2005};\footnote{See also \citet[p.~386]{MTW1973}.} at least one other takes it to be necessary \citet[p.~352]{Knox2013-KNOESG}.

Despite its ubiquity, this common definition of minimal coupling is just as commonly renounced (and often by the same authors), ``because the phrase `the laws of special relativity' is ambiguous. \ldots In general, since the curvature vanishes in special relativity, we can add arbitrary terms involving curvature to our laws'' \citep[p.~201]{Friedman1983}.
Whether such terms are added changes the result of the minimal coupling prescription. (Recall, of course, P-$\xi$ theories.)

Some authors attempt to avoid this problem by formulating minimal coupling as the property that no terms involving the curvature, or second derivatives, appear in matter field equations. 
For example, in place of the substitution-based definition of minimal coupling, \citet[p.~267]{brown2001origin} write that it should be understood as the ``prohibition of curvature coupling in the non-gravitational equations.''\footnote{See also \citet[p.~170]{Brown2005}, \citet[p.~352]{Knox2013-KNOESG}, \citet[p.~22]{schild1967lectures}, and \citet[p.~46]{ehlers1973survey} for similar remarks.} 
But this reformulation is a nonstarter, for reasons explained by 
\citet{BrownRead2016}: ``In GR, second- and higher-order dynamical equations for non-gravitational interactions, in local Lorentz coordinates at a point $p$ and constructed using the minimal coupling scheme, take a form in which the curvature tensor and/or its contractions appear.''

In response to these worries, Brown and Read suggest that one should just accept the ambiguity of the minimal coupling procedure for dynamical equations for matter above first order, but avoid problems by restricting minimal coupling to first-order equations.
(See also \citet[p.~46]{ehlers1973survey}.)  But in fact, the situation is worse than they suggest, for the ambiguities persist even for equations in which one derivative operator appears per term. The reason for this is that, as we have seen, higher order equations can generally be rewritten as first-order equations by introducing auxiliary fields.\footnote{Consider, for instance, the sourced Maxwell equations: $\nabla_a F^{ab} = J^b$ and $\nabla_{[a}F_{bc]}=\mathbf{0}$.  As is well known, these imply that, at least locally, there exists a field $A_a$ such that $F_{ab}=\nabla_{[a}A_{b]}$; and $F_{ab}$ so defined satisfies the Maxwell equations if (and only if) for all such $A_a$, $\nabla_a\nabla^a A^b - \nabla_a\nabla^bA^a = 2J^b$.  But now consider that this latter equation is equivalent to (by permuting the derivatives in the second term) $\nabla_a\nabla^a A^b -R_n{}^bA^n+\nabla^b\nabla_a A^a=2J^b$.  Define a field $G_{ab} = \nabla_a A_b$.  We may then express this most recent equation as a system of first-order equations: $\nabla_a G^{ab} + \nabla^bG^a{}_a = R_n{}^bA^n + 2J^b$; $\nabla_a A_b = G_{ab}$; $\nabla_{[a} G_{bc]} = \mathbf{0}$.  And in this final system, a curvature term appears---even though these equations have the same solutions as the minimally coupled Maxwell equations (up to a choice of gauge), since given any $F_{ab}$ satisfying Maxwell's equations, one can find a field $A_a$ such that $F_{ab}=\nabla_{[a}A_{b]}$ and $(A_a, \nabla_a A_b)$ satisfies these alternative equations; and given any solutions $(A_a,G_{ab})$ to these equations, $F_{ab}=G_{[ab]}$ satisfies Maxwell's equations.  Similar arguments could be made using the wave equation for $F_{ab}$ as discussed by \citet[\S 2.2]{ReadBrownLehmkuhl2018} and introducing an auxiliary field. \label{fn:rewriting}}    Conversely, many apparently first-order equations in special relativity can be given an equivalent form involving higher-order derivatives by adding terms with appropriate indicies containing expressions like $\nabla_{[a}\nabla_{b]}\xi^c$ that always vanish in flat spacetimes not but in general in curve spacetimes.

So, the original problem is in fact more pervasive and pernicious than generally recognized.
The core of it is that the minimal coupling prescription varies with the formulation of the matter theories rather than the theories themselves. It is, once again, hyperintensional.  This, of course, was also our concern about the syntactic criterion for when a theory is LSR---one that we attempted to address by introducing alternative criteria that involved properties of the solutions of a system of equations, rather than their formulation.  

We suggest that this discussion motivates a reconsideration of minimal coupling---one that captures the idea that a system of equations ``depends'' only on the metric and its first derivatives in an intrinsic, geometrical way, rather than by appealing to features of the formulation of the theory.\footnote{
    By a ``geometrical'' way, we mean to refer to the type of mathematics used to formalize the notion of dependence, not to the type of explanatory relation that holds between matter and physical geometry, as is often the concern of \citet{Brown2005}, \citet{Knox2013-KNOESG}, and \citet{ReadBrownLehmkuhl2018}.
}  
One might then hope to prove a theorem of the following form: given a suitably geometrical, or ``semantic'', characterization of minimal coupling, if a matter theory is minimally coupled, then it is LSR by one of the Ehlers criteria.
(The converse would also be of interest, though it is less clear that we should expect it to hold.)  
But of course, without such a definition of ``minimal coupling'' at hand, actually formulating and proving this result is impossible, and so we defer this, too, to future work.





\section*{Acknowledgments}
SCF acknowledges the support of a Marie Curie fellowship (FP7-MC-IIF-628533) during the early development of this project, and helpful feedback from audiences in T\"ubingen (3rd International Interdisciplinary Summer School), London (Sigma Club), Munich (MCMP), Vienna (Center for Quantum Science and Technology), Bucharest (Philosophy Dept.), Salzburg (Philosophy Dept.), and Dubrovnik (42nd Annual Philosophy of Science Conference) on an ancestral version entitled, ``On the Local Flatness of Spacetime.'' JOW: This material is partially based upon work produced for the project “New Directions in Philosophy of Cosmology”, funded by the John Templeton Foundation under grant number 61048; and partly upon work supported by the National Science Foundation under Grant No.~1331126.  We are grateful to David Malament amd two anonymous reviewers for helpful comments on a previous draft and to Thomas Barrett and JB Manchak for discussions of related material.

\singlespacing

\bibliographystyle{elsarticle-harv}
\bibliography{flatness}

\end{document}